\documentclass[12pt,reqno]{amsart}
\usepackage{amsaddr}
\usepackage{a4wide}
\usepackage[dvips]{epsfig}
\usepackage{amscd}
\usepackage{amssymb}
\usepackage{amsthm}
\usepackage{amsmath}
\usepackage{latexsym}
\usepackage[linktocpage,
                  menucolor=blue,
                  linkcolor=blue,
                  citecolor=red,
                  colorlinks=true]{hyperref}

\theoremstyle{plain}
\newtheorem{theorem}{Theorem}[section]
\theoremstyle{plain}

\theoremstyle{definition}

{%
\setcounter{enumi}{0}

\begin{enumerate}}%
{\end{enumerate} }
{%
\setcounter{enumi}{0}

\begin{enumerate}}%
{\end{enumerate} }

\numberwithin{equation}{section}
 \allowdisplaybreaks

\title[]
 {Optimal asset allocation for a DC plan with partial information under inflation and mortality risks }

\begin{document}
\author{Calisto Guambe$^{1,2}$, Rodwell Kufakunesu$^1$, Gusti Van Zyl$^1$, Conrad Beyers$^2$}
\address{$^1$ Department of Mathematics and Applied Mathematics, University of Pretoria, 0002, South Africa }
\address{$^2$ Department of Actuarial Science, University of Pretoria, 0002, South Africa}

\email{calistoguambe@yahoo.com.br, rodwell.kufakunesu@up.ac.za}
\email{gusti.vanzyl@up.ac.za, conrad.beyers@up.ac.za}

\keywords{DC pension plan, Stochastic interest rate, Maximum principle, Stochastic income, inflation risks, mortality risks.}

\begin{abstract}
We study an asset allocation stochastic problem with restriction for a defined-contribution pension plan during the accumulation phase. We consider a financial market with stochastic interest rate, composed of a risk-free asset, a real zero coupon bond price, the inflation-linked bond and the risky asset. A plan member aims to maximize the expected power utility derived from the terminal wealth. In order to protect the rights of a member who dies before retirement, we introduce a clause which allows to withdraw his premiums and the difference is distributed among the survival members. Besides the mortality risk, the fund manager takes into account the salary and the inflation risks. We then obtain closed form solutions for the asset allocation problem using a sufficient maximum principle approach for the problem with partial information. Finally, we give a numerical example
\end{abstract}

\maketitle
\section{Introduction}

Pension funds asset allocation problem has become a very important area of research in recent years. This is motivated by different reasons; for instance, the average age of the employees when they join a pension plan and their expected life time have increased in the last decade. In the area of pension funds, we distinguish two types of pension plans: a Defined Benefit (DB) plan, where the benefits are known in advance and the contributions are adjusted in time to ensure that the fund remains in balance and a Defined Contribution (DC) plan, where the contributions are defined in advance and the benefits depend on the return of the fund, with the risks taken by the plan members. We refer to Antolin {\it et. al.} \cite{antolin} or Devolder {\it et. al.} \cite{devolder} for a thorough discussion on the theory of pension funds. Since most of the developed and developing countries, have moved or are moving from DB to DC plans, where the employee is directly exposed to the financial risks, the study of optimization problems in the context of pension funds it is of great importance. This is because the solution of such problems will help both the pension plan members and the pension fund managers in their allocation of funds in different assets in order to achieve the best retirement savings, even during the periods of market fluctuations or lack of information.

There is a vast of literature dealing with optimization of pension funds problem, for instance, under the expected utility maximization framework, Sun {\it et. al.} \cite{sun2017}, consider a robust portfolio choice for DC pension plan with stochastic income and interest rate. Sun {\it et. al.} \cite{sun2016} study the jump diffusion case of a DC investment plan. Osu {\it et. al.} \cite{osu2017} studied the effect of stochastic extra contribution on DC pension funds, and references therein. This problem has also been considered in the mean variance framework, see, e.g., He and Liang \cite{he2013} and references therein. All the above references solved the DC pension fund problem using a dynamic programming approach under the setting of complete information. Otherwise, Battocchio and Menoncin \cite{battocchio2004} considered a stochastic martingale approach for a DC investment problem. Chen and Delong \cite{chen2015}, studied a DC pension fund problem with regime switching using the techniques of backward stochastic differential equations with quadratic growth.

To the best of our knowledge, in almost all the literature on DC investment problems, the partial information case in the control has not been considered. However, like other investment problems, in the pension fund investment problem, the information about the state control is not always available on time of the decision, which leads to a delayed information about the investment strategy. Thus, one needs to consider the case of DC investments with partial information. We assume that the investment strategy is adapted to a given sub-filtration of the filtration generated by the underlying diffusion processes. Therefore, the dynamic programming approach is not applicable. We use a sufficient maximum principle for such a DC investment problem with stochastic interest rate. In the literature, this method has been widely studied. See, for instance, An and \O ksendal \cite{An-oksendal}, Baghery and \O ksendal \cite{baghery2007}, Framstard {\it et. al.} \cite{framstard} and references therein.

In this paper, we study an asset allocation stochastic problem for a defined-contribution pension plan during the accumulation phase. We consider a financial market with stochastic interest rate, composed by a risk-free asset, a real zero coupon bond price, the inflation-linked bond and the risky asset, where a plan member aims to maximize the expected power utility derived from the terminal wealth. In order to protect the rights of a member who dies before retirement, we introduce a clause which allows the member to withdraw his premiums and the difference is distributed among the survival members. Besides the mortality risk, the fund manager takes into account the salary and the inflation risks. Furthermore, due to the ultimate aim of the pension fund and to prevent the members from losing all their savings, we introduce a restriction in their investment choices. This restriction forces the plan members to put a certain proportion of their savings in a safe investment.

This paper unifies most of the features considered on DC investment problems, such as, stochastic interest rate, inflation risks, mortality risks, stochastic income, etc, and study the optimization problem under partial information case.

The rest of the paper is organized as follows: in Section \ref{sec:modelformulation}, we introduce the setting assumptions of the financial market, namely, the stochastic interest rate, the inflation linked asset, the zero coupon bond price and the risky asset. We also consider the existence of stochastic income and we state the main optimization problem under study. In Section \ref{sec:solutionofpension}, we solve the asset allocation problem of the pension fund manager with partial information using the maximum principle approach presented in the Appendix. Finally, we give a numerical example in Section \ref{sec:numericalexample}.

\section{The model formulation}\label{sec:modelformulation}
Consider three independent Brownian motions $\{W_r(t);W_I(t);W_S(t),\, 0\leq t\leq T\}$ associated to the complete filtered probability space $(\Omega,\mathcal{F},\{\mathcal{F}_t\},\mathbb{P})$. Let a fixed horizon investment of a defined contribution pension fund, with a retirement date denoted by $T<\infty.$ Since we are dealing with a long term investment (pension funds), it is reasonable to consider a stochastic interest rate. Thus, we assume that the interest rate $r(t)$ satisfies the following stochastic differential equation (SDE)
\begin{equation}\label{interestrate}
    dr(t)=a(\bar{r}-r(t))dt+\sigma_rdW_r(t)\,,
\end{equation}
where $a,\,\bar{r}$ and $\sigma_r$ are positive constants, with $a$ representing the level of mean reverting, $\bar{r}$ is the long-run mean of interest rate and $\sigma_r$ the volatility.

Given the stochastic interest rate, we can derive the value of the zero coupon bond in order to hedge the fluctuations of the interest rate. Its price is given by
$$
P(t,T):=\exp\Bigl\{-\int_t^Tr(s)ds\Bigl\}\,.
$$
Applying the It\^o's formula, we obtain:
\begin{equation}\label{couponbond}
    dP(t,T)=P(t,T)\Bigl[\Bigl(r(t)+\frac{\sigma_r}{a}\xi\left(1-e^{-a(T-t)}\right)\Bigl)dt -\frac{\sigma_r}{a}\left(1-e^{-a(T-t)}\right)dW_r(t)\Bigl]\,.
\end{equation}

In order to capture the inflation risks, we also consider an inflation index $I(t)$ given by
$$
dI(t)=I(t)[\mu_I(t)dt+\sigma_I(t)dW_I(t)]\,,
$$
with deterministic expected rate of inflation $\mu_I$ and  volatility $\sigma_I$ satisfying the following integrability condition.
$$
\int_0^T\left[|\mu_I(t)|+\sigma_I^2(t)\right]dt<\infty, \ \ \ \rm{a.s.}
$$
The inflation-linked bond price is defined by
$$
B(t)=I(t)S_0(t)\,,
$$
where $S_0$ is the risk-free asset price. Then,
\begin{equation}\label{inflationbond}
    dB(t)=B(t)[(r(t)+\mu_I(t))dt+\sigma_I(t)dW_I(t)]\,.
\end{equation}

Finally, assume that the pension member also allocate the funds in a risky asset defined by the following geometric diffusion process
$$
dS(t)=S(t)\left[\mu_S(t)dt+\sigma(t)dW_S(t)+\sigma_S(t)dW_r(t)\right]\,,
$$
where the mean rate of return $\mu_S(t):=r(t)+\mu(t)$, the volatilities $\sigma(t),\,\sigma_S(t)$ are deterministic functions, satisfying the following integrability condition
\begin{equation}\label{integrabilityS}
\int_0^T\left[|\mu_S(t)|+\sigma^2(t)+\sigma_S^2(t)\right]dt<\infty, \ \ \ \rm{a.s.}
\end{equation}

We suppose that a pension member has a stochastic income salary driven by:

\begin{equation}\label{incomesalary}
    d\ell(t)=\ell(t)[(\mu_\ell(t)+r(t))dt+\sigma_1(t)dW_r(t)+\sigma_2(t)dW_S(t)]\,,
\end{equation}
where $\mu_\ell(t)+r(t)$ is the expected growth rate of income, $\mu_\ell,\,\sigma_1$ and $\sigma_2$ are deterministic functions also satisfying the integrability condition as in \eqref{integrabilityS}.

Moreover, suppose that the pension member contributes an amount of $\delta\ell(t)$, at time $t$, where $\delta\in(0,1)$ is the proportion of the salary contributed to the pension plan. We assume that the accumulation period of the fund starts from age $t_0>0$ of the member, until the retirement age $t_0+T$. In order to protect the rights of the plan members who die before retirement, we adopt the withdrawal of the premiums for the member who dies, as in He and Liang \cite{he2013}.

Let $M_0$ be the number of members who are still alive in the pension at time $t$. Then, the expected number of members who will die during the time interval $(t,t+\delta t)$ is $M_0P_{t_0+t}$, where $P_{t_0+t}$ is the probability of a person who is alive at the age $t_0+t$ and will die in the following $\delta t$ period. The expected number of members who are actually alive at time $t+\delta t$ is $M_0(1-P_{t_0+t})$, which is a deterministic function of time.

Based on He and Liang \cite{he2013}, we adopt the De Moivre mortality model, i.e., the deterministic force of mortality $\beta_{t_0}(t)=\frac{1}{\tau+t_0-t}$\,. Then,
$$
P_{t_0+t}=\exp\Bigl\{-\int_{t_0+t}^{t_0+s}\beta_{t_0}(u)du\Bigl\}=\frac{\tau-s}{\tau-t}\,, \ \ \ \ s>t\,.
$$

We consider a sub-filtration
$$
\mathcal{E}_t\subseteq\mathcal{F}_t\,, \ \ \ \forall \,t\in[0,T]\,,
$$
where $\mathcal{E}_t$ represents the amount of the information available to the pension manager at time $t$.

Since we are modeling an investment plan for pension funds, we assume that there is a proportion of the pension members wealth restricted to a safe investment (risk-free asset). We denote that fraction by $\kappa$. Let $\pi_1(t),\,\pi_2(t),\,\pi_3(t)$ be the $\{\mathcal{E}_t\}_{t\in[0,T]}$-adapted processes denoting the proportions of the wealth invested in the inflation-linked bond, zero coupon bond and a risky asset respectively. Then $1-\kappa-\pi_1(t)-\pi_2(t)-\pi_3(t)\in\mathcal{E}_t$ is the proportion of the wealth invested in a risk-free asset. The wealth process is then given by
\begin{eqnarray}\nonumber
  dX(t) &=& \Bigl[X(t)\Bigl((1-\kappa) r(t)+\mu_I(t)\pi_1(t)+\frac{\sigma_r}{a}\xi\Bigl(1-e^{-a(T-t)}\Bigl)\pi_2(t)+\mu_S(t)\pi_3(t) +\beta_{t_0}(t)\Bigl) \\ \nonumber
   && +(1-\varepsilon t\beta_{t_0})\delta\ell(t)\Bigl]dt +\pi_1(t)\sigma_I(t)X(t)dW_I(t) \\
   &&   +\Bigl[\pi_3(t)\sigma_S(t)-\frac{\sigma_r}{a}\Bigl(1-e^{-a(T-t)}\Bigl)\pi_2(t)\Bigl]X(t)dW_r(t) +\pi_3(t)\sigma(t)X(t)dW_S(t)\,.
\end{eqnarray}
Here $\varepsilon$ is a parameter with values $0$ or $1$. If $\varepsilon=0$, the pension member obtains nothing during the accumulation phase, while if $\varepsilon=1$, the premiums are returned to the member when he dies.

Suppose that the income salary $\ell(t)$ is given as a numeraire. we define the relative wealth process by $Y(t)=\frac{X(t)}{\ell(t)}$. Then by It\^o's formula, we have

\begin{eqnarray}\label{relativewealth}
 && dY(t) \\ \nonumber
   &=& \Bigl\{Y(t)\Bigl[\mu_I(t)\pi_1(t) +\frac{\sigma_r}{a}\xi\Bigl(1-e^{-a(T-t)}\Bigl)\pi_2(t)+\mu_S(t)\pi_3(t)-\kappa r(t) +\beta_{t_0}(t) -\mu_\ell(t) \\ \nonumber
   && +(\sigma_1^2(t)+\sigma_2^2(t))-\pi_3(t)\sigma_S(t)\sigma_2(t) - \Bigl(\pi_3(t)\sigma_S(t)-\frac{\sigma_r}{a}\Bigl(1-e^{-a(T-t)}\Bigl)\pi_2(t)\Bigl)\sigma_1(t)\Bigl] \\ \nonumber
   &&  -(1-\varepsilon t\beta_{t_0})\delta\Bigl\}dt +\pi_1(t)\sigma_I(t)Y(t)dW_I(t) +(\pi_3(t)\sigma(t)-\sigma_2(t))Y(t)dW_S(t) \\ \nonumber
   && +\Bigl[\pi_3(t)\sigma_S(t)-\frac{\sigma_r}{a}\Bigl(1-e^{-a(T-t)}\Bigl)\pi_2(t)-\sigma_1(t)\Bigl]Y(t)dW_r(t)\,.
\end{eqnarray}

Define $\mathcal{A}:=\{(\pi_1,\pi_2,\pi_3):=(\pi_1(t),\pi_2(t),\pi_3(t))_{t\in[0,T]}\}$ as a set of admissible strategies if $(\pi_1(t),\pi_2(t),\pi_3(t))\in\{\mathcal{E}_t\}_{t\in[0,T]}$ and the SDE \eqref{relativewealth} has a unique strong solution such that $X(t)\geq0$, $\mathbb{P}$-a.s.

Let $U:(0,\infty)\mapsto\mathbb{R}$ be the utility function measuring the investor's preference. The main objective of the pension fund manager is to maximize the following functional:
\begin{equation*}
    \mathcal{J}(t,r,y,\pi_1,\pi_2,\pi_3)=\mathbb{E}_{t,x,r}[U(Y(T))]\,.
\end{equation*}
Then, the value function of the pension manager is given by
\begin{equation}\label{valueproblem}
V(t,r,y)=\sup_{(\pi_1,\pi_2,\pi_3)\in\mathcal{A}}\mathcal{J}(t,r,y,\pi_1,\pi_2,\pi_3)\,.
\end{equation}

\section{Solution of the pension fund manager optimization problem} \label{sec:solutionofpension}

Since we consider an asset allocation problem with partial information, the classical dynamic programming approach applied, for instance, in Battocchio and Menoncin \cite{battocchio2004}, Federico \cite{federico2008}, Di Giacinto {\it et. al.} \cite{di2011}, Sun {\it et. al.} \cite{sun2017} is not applicable.

Applying a sufficient maximum principle approach for diffusion stochastic volatility model with partial information (see the results in the Appendix), we define the Hamiltonian $\mathcal{H}:[0,T]\times \mathbb{R}\times\mathbb{R}\times(0,1)^3\times\mathbb{R}\times\mathbb{R}\times\mathbb{R}\times\mathbb{R} \times\mathbb{R}\times\mathbb{R}\rightarrow\mathbb{R}$ by:

\begin{eqnarray}\label{hamiltonian}
  && \mathcal{H}(t,r(t),Y(t),\pi_1(t),\pi_2(t),\pi_3(t),A_1(t),A_2(t),B_1(t),B_2(t), B_3(t),B_4(t))  \\ \nonumber
  &=& \Bigl\{Y(t)\Bigl[\mu_I(t)\pi_1(t) +\frac{\sigma_r}{a}\xi\Bigl(1-e^{-a(T-t)}\Bigl)\pi_2(t)+\mu_S(t)\pi_3(t)-\kappa r(t) +\beta_{t_0}(t) -\mu_\ell(t)   \\ \nonumber
  && +\Bigl(\pi_3(t)\sigma_S(t)-\frac{\sigma_r}{a}\Bigl(1-e^{-a(T-t)}\Bigl)\pi_2(t)\Bigl)\sigma_1(t)\Bigl] -(1-\varepsilon t\beta_{t_0})\delta\Bigl\}A_1(t) \\ \nonumber
   && + a(\bar{r}-r(t))A_2(t) +\pi_1(t)\sigma_I(t)Y(t)B_1(t) +(\pi_3(t)\sigma(t)-\sigma_2(t))Y(t)B_3(t) \\ \nonumber
  && + \Bigl[\pi_3(t)\sigma_S(t)-\frac{\sigma_r}{a}\Bigl(1-e^{-a(T-t)}\Bigl)\pi_2(t)-\sigma_1(t)\Bigl]Y(t)B_2(t) +\sigma_rB_4(t)\,.
\end{eqnarray}
The adjoint equations corresponding to the admissible strategy $(\pi_1,\pi_2,\pi_3)$ are given by the following backward stochastic differential equations

\begin{eqnarray}\nonumber
  dA_1(t) &=& -\Bigl\{\mu_I(t)\pi_1(t) +\frac{\sigma_r}{a}\xi\Bigl(1-e^{-a(T-t)}\Bigl)\pi_2(t)+\mu_S(t)\pi_3(t)-\kappa r(t)  \\ \nonumber
   && +\beta_{t_0}(t) -\mu_\ell(t) +\sigma_1^2(t) +\sigma_2^2(t)-\pi_3(t)\sigma_S(t)\sigma_2(t) \\ \nonumber
&& - \Bigl(\pi_3(t)\sigma_S(t)-\frac{\sigma_r}{a}\Bigl(1-e^{-a(T-t)}\Bigl)\pi_2(t)\Bigl)\sigma_1(t)\Bigl]A_1(t) \\ \nonumber
   &&  +\pi_1(t)\sigma_I(t)B_1(t) +(\pi_3(t)\sigma(t)-\sigma_2(t))B_3(t) \\ \nonumber
   && + \Bigl[\pi_3(t)\sigma_S(t)-\frac{\sigma_r}{a}\Bigl(1-e^{-a(T-t)}\Bigl)\pi_2(t)-\sigma_1(t)\Bigl]B_2(t)\Bigl\}dt \\ \label{adjointa1}
   && +B_1(t)dW_r(t)+B_2(t)dW_I(t)+B_3(t)dW_S(t) \\ \nonumber
   A_1(T) &=& U'(Y(T))
\end{eqnarray}
and

\begin{eqnarray}\label{adjointa2}
  dA_2(t) &=& [\kappa Y(t)A_1(t)+aA_2(t)]dt +B_4(t)dW_r(t) +B_5(t)dW_I(t) +B_6(t)dW_S(t) \\ \nonumber
  A_2(T) &=& 0\,.
\end{eqnarray}

Applying the first order conditions of optimality to the Hamiltonian with respect to $(\pi_1,\pi_2,\pi_3)$, given the information available $\{\mathcal{E}_t\}_{t\in[0,T]}$, we have the following equations

\begin{equation}\label{firstorderconditions}
\left\{
  \begin{array}{ll}
    \mu_I(t)\mathbb{E}[A_1^*(t)\mid\mathcal{E}_t] +\sigma_I(t)\mathbb{E}[B_1^*(t)\mid\mathcal{E}_t] \,=\, 0, &  \\
    (\xi+\sigma_1(t))\mathbb{E}[A_1^*(t)\mid\mathcal{E}_t]-\mathbb{E}[B_2^*(t)\mid\mathcal{E}_t] \,=\, 0, &  \\
    (\mu_S(t)-\sigma_S(t)(\sigma_1(t)+\sigma_2(t)))\mathbb{E}[A_1^*(t)\mid\mathcal{E}_t] +\sigma_S(t)\mathbb{E}[B_2^*(t)\mid\mathcal{E}_t]+\sigma(t)\mathbb{E}[B_3^*(t)\mid\mathcal{E}_t] \,=\, 0, &
  \end{array}
\right.
\end{equation}
where $A_1^*,B_1^*,B_2^*$ and $B_3^*$ are the adjoint processes corresponding to the optimal controls $(\pi_1^*,\pi_2^*,\pi_3^*)$. For this optimal controls, the adjoint equations become

\begin{eqnarray}\nonumber
  dA_1^*(t) &=& -\Bigl\{[\beta_{t_0}(t)-\kappa r(t) -\mu_\ell(t) +\sigma_1^2(t) +\sigma_2^2(t)]A_1^*(t) -\sigma_1(t)B_2^*(t) -\sigma_2(t)B_3^*(t)\Bigl\}dt \\  \label{adjointa1}
   && +B_1^*(t)dW_r(t)+B_2^*(t)dW_I(t)+B_3^*(t)dW_S(t) \\ \nonumber
   A_1^*(T) &=& U'(Y(T))
\end{eqnarray}
and

\begin{eqnarray}\label{adjointa2}
  dA_2^*(t) &=& [\kappa Y(t)A_1^*(t)+aA_2^*(t)]dt +B_4^*(t)dW_r(t) +B_5^*(t)dW_I(t) +B_6^*(t)dW_S(t) \\ \nonumber
  A_2^*(T) &=& 0\,.
\end{eqnarray}

In order to solve our optimization problem, we consider a power utility function of the form $U(y)=\frac{y^\alpha}{\alpha}$, where $\alpha\in(-\infty,1)\setminus\{0\}$. Then the terminal condition for the first adjoint equation becomes $A_1(T)=Y(T)^{\alpha-1}$. From this form, we try the solution of the BSDE \eqref{adjointa1} to be of the form
\begin{equation}\label{optimala1}
A_1^*(t)=(Y(t))^{\alpha-1}e^{\varphi(t)+\phi(t)r}\,, \ \ \ \varphi(T)=\phi(T)=0\,.
\end{equation}
Applying the It\^o's formula, we have
\begin{eqnarray*}
 && \frac{dA_1^*(t)}{A_1^*(t)} \\
  &=& \Bigl\{\varphi'(t) +\phi'(t)r+a(\bar{r}-r)\phi(t) +\frac{1}{2}\sigma_r^2(\phi(t))^2 \\
  && +\frac{1}{2}(\alpha-1)\sigma_r\Bigl[\pi_3(t)\sigma_S(t)-\frac{\sigma_r}{a}\Bigl(1-e^{-a(T-t)}\Bigl)\pi_2(t)-\sigma_1(t)\Bigl]\phi(t) \\
  &&  +(\alpha-1)\Bigl[\mu_I(t)\pi_1^*(t)+\mu_S(t)\pi_3^*(t) +\frac{\sigma_r}{a}\xi\Bigl(1-e^{-a(T-t)}\Bigl)\pi_2^*(t) \\
   && -\kappa r(t) +\beta_{t_0}(t)-\mu_\ell(t)+\sigma_1^2(t)+\sigma_2^2(t)-\sigma_S(t)\sigma_2(t)\pi_3^*(t) +\frac{1}{2}(\alpha-2)(\pi_1^*(t))^2\sigma_I^2(t)  \\
   &&   -\Bigl(\sigma_S(t)\pi_3^*(t)-\frac{\sigma_r}{a}\xi\Bigl(1-e^{-a(T-t)}\Bigl)\pi_2^*(t)\Bigl)\sigma_1(t) +\frac{1}{2}(\alpha-2)(\sigma_S(t)\pi_3^*(t)-\sigma_2(t))^2 \\
   && +\frac{1}{2}(\alpha-2)\Bigl(\sigma_S(t)\pi_3^*(t)-\frac{\sigma_r}{a}\xi\Bigl(1-e^{-a(T-t)}\Bigl)\pi_2^*(t)-\sigma_1(t)\Bigl)^2 \\
   && -(1-\varepsilon t\beta_{t_0}(t))\delta (y(t))^{-1}\Bigl]\Bigl\}dt +(\alpha-1)\pi_1^*(t)\sigma_I(t)dW_I(t) \\
   && +\Bigl[(\alpha-1) \Bigl(\sigma_S(t)\pi_3^*(t)-\frac{\sigma_r}{a}\xi\Bigl(1-e^{-a(T-t)}\Bigl)\pi_2^*(t)-\sigma_1(t)\Bigl)+\sigma_r\phi(t)\Bigl]dW_r(t) \\
   && +(\alpha-1)(\sigma_S(t)\pi_3^*(t)-\sigma_2(t))dW_S(t)\,.
\end{eqnarray*}
Comparing with the adjoint equation \eqref{adjointa1}, we obtain the following relations
\begin{eqnarray}\label{optimalb1}
  B_1^*(t) &=& \Bigl[(\alpha-1) \Bigl(\sigma_S(t)\pi_3^*(t)-\frac{\sigma_r}{a}\xi\Bigl(1-e^{-a(T-t)}\Bigl) \pi_2^*(t)-\sigma_1(t)\Bigl) \\ \nonumber
  && \ \ \ \ \ \ \ +\sigma_r\phi(t)\Bigl]A_1^*(t)\,; \\ \label{optimalb2}
  B_2^*(t) &=& (\alpha-1)\pi_1^*(t)\sigma_I(t)A_1^*(t)\,; \\ \label{optimalb3}
  B_3^*(t) &=&  (\alpha-1)(\sigma_S(t)\pi_3^*(t)-\sigma_2(t))A_1^*(t)\,.
\end{eqnarray}
Moreover, $\varphi(t)$ and $\phi(t)$ solve the following backward ordinary differential equation

\begin{eqnarray}\label{backwardode}
    && \Bigl\{\varphi'(t)+r\phi'(t) +\frac{1}{2}\sigma_r^2(\phi(t))^2 +K(t) \\ \nonumber
&& + \Bigl[a(\bar{r}-r)+ \frac{1}{2}(\alpha-1)\sigma_r\Bigl(\pi_3(t)\sigma_S(t)-\frac{\sigma_r}{a}\Bigl(1-e^{-a(T-t)}\Bigl)\pi_2(t) -\sigma_1(t)\Bigl)\Bigl]\phi(t)\Bigl\}A_1^*(t) \\ \nonumber
&=&    Q(t)A_1^*(t)+\sigma_1(t)B_2^*(t)+\sigma_2(t)B_3^*(t)\,,
\end{eqnarray}
where
\begin{eqnarray*}
 && K(t) \\
&=& (\alpha-1)\Bigl[\mu_I(t)\pi_1^*(t)+\mu(t)\pi_3^*(t) +\frac{\sigma_r}{a}\xi\Bigl(1-e^{-a(T-t)}\Bigl)\pi_2^*(t) -\kappa r(t) +\beta_{t_0}(t)-\mu_\ell(t) \\
   && +\sigma_1^2(t)+\sigma_2^2(t)-\sigma_S(t)\sigma_2(t)\pi_3^*(t) +\frac{1}{2}(\alpha-2)(\pi_1^*(t))^2\sigma_I^2(t)  \\
   &&   -\Bigl(\sigma_S(t)\pi_3^*(t)-\frac{\sigma_r}{a}\xi\Bigl(1-e^{-a(T-t)}\Bigl)\pi_2^*(t)\Bigl)\sigma_1(t) +\frac{1}{2}(\alpha-2)(\sigma_S(t)\pi_3^*(t)-\sigma_2(t))^2 \\
   && +\frac{1}{2}(\alpha-2)\Bigl(\sigma_S(t)\pi_3^*(t)-\frac{\sigma_r}{a}\xi\Bigl(1-e^{-a(T-t)}\Bigl)\pi_2^*(t)-\sigma_1(t)\Bigl)^2 -(1-\varepsilon t\beta_{t_0}(t))\delta (y(t))^{-1}\Bigl]
\end{eqnarray*}
and
\begin{equation*}
    Q(t)= \kappa r(t) +\mu_\ell(t)-\beta_{t_0}(t) -\sigma_1^2(t) -\sigma_2^2(t)\,.
\end{equation*}
Substituting \eqref{optimala1}, \eqref{optimalb1}--\eqref{optimalb3} into \eqref{firstorderconditions}, we obtain the following optimal solutions:

\begin{eqnarray}\label{optimalpi1}
    \pi_1^*(t) &=& \frac{\xi+\sigma_1(t)}{(\alpha-1)\sigma_I(t)}\,; \\ \label{optimalpi2}
\pi_2^*(t) &=& \frac{2a}{(\alpha-1)\sigma_r\sigma_I(t)\xi\Bigl(1-e^{-a(T-t)}\Bigl)} \Bigl[\sigma_r\sigma_I(t)\phi(t)+\mu_I(t) \\ \nonumber
&& (\alpha-1)\sigma_I(t)(\sigma_S(t)\pi_3^*(t)-\sigma_1(t))\Bigl]\,; \\ \label{optimalpi3}
\pi_3^*(t) &=& \frac{\mu_S(t)-\sigma_S(t)\sigma_2(t)+\xi\sigma_S(t)+ (1-\alpha)\sigma(t)\sigma_2(t)}{(1-\alpha)\sigma(t)\sigma_S(t)}\,.
\end{eqnarray}

We point out that from \eqref{optimala1}, \eqref{optimalb2} and \eqref{optimalb3}, we can write \eqref{backwardode} in the following system

$$
\left\{
  \begin{array}{ll}
    \varphi'(t)+\mathcal{K}(t) & \hbox{=\,\,0\,;} \\
    r\phi'(t)+M(t)\phi(t)+\frac{1}{2}\sigma_r^2(\phi(t))^2 & \hbox{=\,\,0\,,}
  \end{array}
\right.
$$
where
\begin{eqnarray*}
    \mathcal{K}(t) &=& K(t)-[Q(t)+(\alpha-1)(\pi_1^*(t)\sigma_1(t)\sigma_I(t) +\sigma_2(t)(\pi_3^*(t)\sigma_S(t)-\sigma_2(t)))]\,, \\
    M(t) &=& a(\bar{r}-r)+ \frac{1}{2}(\alpha-1)\sigma_r\Bigl(\pi^*_3(t)\sigma_S(t)-\frac{\sigma_r}{a}\Bigl(1-e^{-a(T-t)}\Bigl)\pi^*_2(t)-\sigma_1(t)\Bigl)
\end{eqnarray*}
which give the following solutions

\begin{equation*}
    \varphi(t)=-\int_t^T\mathcal{K}(s)ds\,, \ \ \ \ t\in[0,T]\,,
\end{equation*}
and
\begin{equation*}
    \phi(t)=\frac{1}{2}\sigma_r^2\exp\Bigl\{-\int_t^T(r(s))^{-1}M(s)ds\Bigl\}\int_t^T(r(s))^{-1}ds\,.
\end{equation*}
This completes the solution \eqref{optimala1}.

For the second adjoint equation, we have that from the solution of $A_1^*(t)$ in \eqref{optimala1}, we can write
$$
dA_2^*(t)=\Bigl[\kappa h(t)(y(t))^{\alpha}+aA_2^*(t)\Bigl]dt+B_4^*(t)dW_r(t)+B_5^*(t)dW_I(t)+B_6^*(t)dW_S(t)\,,
$$
which is a linear BSDE. From Cohen and Elliott \cite{cohen}, Theorem 19.2.2., $A_2^*(t)$ is given by
$$
A_2^*(t)=-\mathbb{E}\Bigl[\kappa\int_t^Th(s)(y(s))^{\alpha}ds\mid\mathcal{E}_t\Bigl]\,.
$$
The controls $B_4^*,\,B_5^*,\,B_6^*$ can be obtained using the martingale representation theorem.

We then conclude this section summarizing our results in the following theorem.

\begin{theorem}
Under the power utility function, the optimal strategies for a defined contribution problem \eqref{valueproblem}, based on the information flow $\{\mathcal{E}_t\}_{t\in[0,T]}$, are given by

\begin{eqnarray*}
    \pi_1^*(t) &=& \frac{\xi+\sigma_1(t)}{(\alpha-1)\sigma_I(t)}\,; \\
\pi_2^*(t) &=& \frac{2a}{(\alpha-1)\sigma_r\sigma_I(t)\xi\Bigl(1-e^{-a(T-t)}\Bigl)} \Bigl[\sigma_r\sigma_I(t)\phi(t)+\mu_I(t) \\ \nonumber
&& (\alpha-1)\sigma_I(t)(\sigma_S(t)\pi_3^*(t)-\sigma_1(t))\Bigl]\,; \\
\pi_3^*(t) &=& \frac{\mu_S(t)-\sigma_S(t)\sigma_2(t)+\xi\sigma_S(t)+ (1-\alpha)\sigma(t)\sigma_2(t)}{(1-\alpha)\sigma(t)\sigma_S(t)}\,.
\end{eqnarray*}
where
\begin{equation*}
    \phi(t)=\frac{1}{2}\sigma_r^2\exp\Bigl\{-\int_t^T(r(s))^{-1}M(s)ds\Bigl\}\int_t^T(r(s))^{-1}ds\,.
\end{equation*}
\end{theorem}

\section{Numerical example}\label{sec:numericalexample}

In this section, we consider a numerical application of our results, in order to show the behavior of the optimal portfolio strategy derived in the previous section. We assume the following parameters consistent with the numerical analysis in Battocchio and Menoncin \cite{battocchio2004}. The graphs below show that the fund manager should completely not invest in the inflation-linked asset, as that attracts negative interest rate. This is in line with the literature, see, e.g., \cite{beletski2006}, \cite{Basimanebotlhe}. The allocation in the stock follows the behavior of the interest rate, which means that more wealth is invested in the stock as the interest rate increases. For the zero-coupon bond, the graph bellow suggests that a small proportion of the wealth should be invested in this asset, along the life time of the investment.

$$\left(
  \begin{array}{cccccccc}
    a & \bar{r} & \sigma_r & r(0) & T & \xi & \mu_I & \sigma_I \\
    0.2 & 0.05 & 0.02 & 0.03 & 20 & 0.15 & -0.01 & 0.015 \\
    \mu & \sigma & \sigma_S & \mu_\ell & \sigma_1 & \sigma_2 & \ell(0) & \delta \\
    0.06 & 0.19 & 0.06 & 0.01 & 0.014 & 0.171 & 100 & 0.12 \\
  \end{array}
\right)
$$

\begin{figure}[h]
  \includegraphics[width=8cm]{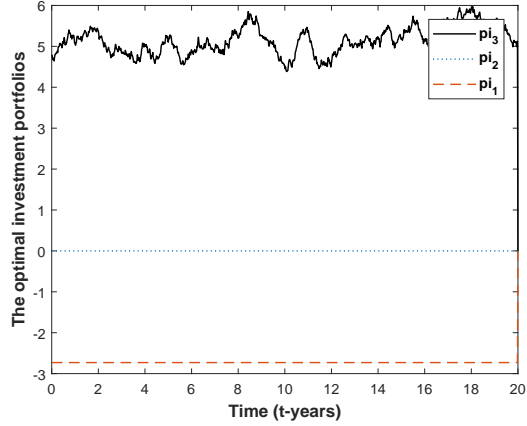}\\
  \caption{The graph shows how the fund manager should pursue with the portfolio allocation for $\alpha=-3$.}\label{portfolio}
\end{figure}

\begin{figure}[h]
  \includegraphics[width=8cm]{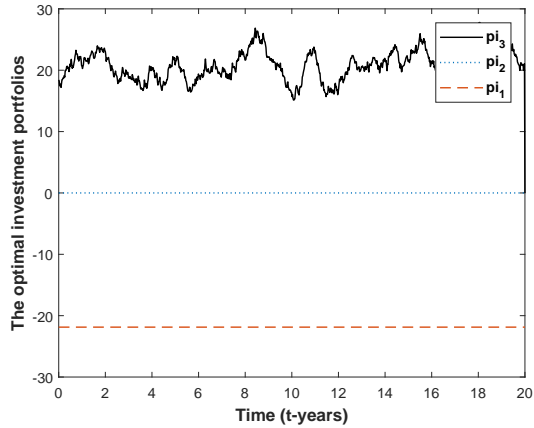}\\
  \caption{The graph shows how the fund manager should pursue with the portfolio allocation for $\alpha=0.5$.}\label{portfolio1}
\end{figure}

\section*{Appendix}
We introduce a version of a maximum principle approach for stochastic volatility model under diffusion with partial information, which is mainly based on the results in Guambe and Kufakunesu \cite{guambe}.
On a complete filtered probability space $(\Omega,\mathcal{F},\{\mathcal{F}_t\}_{t\in[0,T]},\mathbb{P})$, suppose that the dynamics of the state process is given by the following stochastic differential equation (SDE)
\begin{eqnarray}\label{stateprocess}
  dX(t) &=& b(t,X(t),Y(t),\pi(t))dt+ \sigma(t,X(t),Y(t),\pi(t))dW_1(t)  \\ \nonumber
   && +\beta(t,X(t),Y(t),\pi(t))dW_2(t)\,; \\ \nonumber
  X(0) &=& x\in\mathbb{R}\,,
\end{eqnarray}
where the external economic factor $Y$ is given by
\begin{equation}\label{externalgeneral}
    dY(t)=\varphi(Y(t))dt+\phi(Y(t))dW_2(t)\,.
\end{equation}

We assume that the functions $b,\sigma,\beta:[0,T]\times\mathbb{R}\times\mathbb{R}\times\mathcal{A}\rightarrow\mathbb{R}$; $\varphi,\phi:\mathbb{R}\rightarrow\mathbb{R}$ are given predictable processes, such that \eqref{stateprocess} and \eqref{externalgeneral} are well defined and \eqref{stateprocess} has a unique solution for each $\pi\in\mathcal{A}$. Here, $\mathcal{A}$ is a given closed set in $\mathbb{R}$. We assume that the control process $\pi$ is adapted to a given filtration $\{\mathcal{E}_t\}_{t\in[0,T]}$, where
$$
\mathcal{E}_t\subseteq\mathcal{F}_t\,, \ \ \ \forall \,t\in[0,T]\,.
$$
The sub-filtration $\{\mathcal{E}_t\}_{t\in[0,T]}$ denotes the amount of the information available to the controller at time $t$ about the state of the system\,.

Let $f:[0,T]\times\mathbb{R}\times\mathbb{R}\times\mathcal{A}\rightarrow\mathbb{R}$ be a continuous function and $g:\mathbb{R}\times\mathbb{R}\rightarrow\mathbb{R}$ a concave function. We define the performance criterion by
\begin{equation}\label{functionalgeneral}
    \mathcal{J}(t)=\mathbb{E}\Bigl[\int_0^Tf(t,X(t),Y(t),\pi(t))dt+g(X(T),Y(T))\Bigl]\,.
\end{equation}
We say that $\pi\in\mathcal{A}$ is an admissible strategy if \eqref{stateprocess} has a unique strong solution and
$$
\mathbb{E}\Bigl[\int_0^T|f(t,X(t),Y(t),\pi(t))|dt+|g(X(T),Y(T))|\Bigl]<\infty\,.
$$
The partial information control problem is to find $\pi^*\in\mathcal{A}$ such that
$$
\mathcal{J}(\pi^*)=\sup_{\pi\in\mathcal{A}}\mathcal{J}(\pi)\,.
$$
The control $\pi^*$ is called an optimal control if it exists.

In order to solve this stochastic optimal control problem with stochastic volatility, we use the so called maximum principle approach. The beauty of this method is that it solves a stochastic control problem in a more general situation, that is, for both Markovian and non-Markovian cases. We point out that, due to the nature of the partial information $\{\mathcal{E}\}_{t\in[0,T]}$, the dynamic programming approach for a stochastic volatility model by Pham \cite{pham} is not applicable. Our approach may be considered as an extension of the maximum principle approach for a stochastic control problem with partial information in Baghery and \O ksendal  \cite{baghery2007} to the stochastic volatility case.

We define the Hamiltonian $\mathcal{H}:[0,T]\times \mathbb{R}\times\mathbb{R}\times\mathbb{A}\times\mathbb{R}\times\mathbb{R}\times\mathbb{R}\times\mathbb{R} \rightarrow\mathbb{R}$ by:

\begin{eqnarray}\label{hamiltoniangeneral}
  && \mathcal{H}(t,X(t),Y(t),\pi(t),A_1(t),A_2(t),B_1(t),B_2(t))  \\ \nonumber
  &=& f(t,X(t),Y(t),\pi(t))+b(t,X(t),Y(t),\pi(t))A_1(t) +\varphi(Y(t))A_2(t) \\ \nonumber
   && + \sigma(t,X(t),Y(t),\pi(t))B_1(t) +\beta(t,X(t),Y(t),\pi(t))B_2(t)+\phi(Y(t))B_3(t)\,,
\end{eqnarray}
 From now on, we assume that the Hamiltonian $\mathcal{H}$ is continuously differentiable w.r.t. $x$ and $y$. Then, the adjoint equations corresponding to the admissible strategy $\pi\in\mathcal{A}$ are given by the following $\{\mathcal{F}_t\}_{t\in[0,T]}$-adapted backward stochastic differential equations (BSDEs)
\begin{eqnarray}\nonumber
  dA_1(t) &=& -\frac{\partial\mathcal{H}}{\partial x}(t,X(t),Y(t),\pi(t),A_1(t),A_2(t),B_1(t),B_2(t))dt  \\ \label{adjointg1}
   && +B_1(t)dW_1(t) +B_2(t)dW_2(t)\,, \\
  A_1(T) &=& \frac{\partial g}{\partial x}(X(T),Y(T))
\end{eqnarray}
and
\begin{eqnarray}\nonumber
  dA_2(t) &=& -\frac{\partial\mathcal{H}}{\partial y}(t,X(t),Y(t),\pi(t),A_1(t),A_2(t),B_1(t),B_2(t))dt  \\ \label{adjointg2}
   && +B_3(t)dW_1(t) +B_4(t)dW_2(t)\,, \\
  A_2(T) &=& \frac{\partial g}{\partial y}(X(T),Y(T))\,.
\end{eqnarray}

The verification theorem associated to our problem is stated as follows:

\begin{theorem}\label{sufficientmp}
$\rm{ (Sufficient \ maximum \ principle)}$
Let $\pi^*\in\mathcal{A}$ with the corresponding wealth process $X^*$. Suppose that the pairs $(A_1^*(t),B_1^*(t),B_2^*(t))$ and $(A_2^*(t),B_3^*(t),B_4^*(t))$ are the solutions of the adjoint equations \eqref{adjointg1} and \eqref{adjointg2}, respectively. Moreover, suppose that the following inequalities hold:
\begin{itemize}
  \item[(i)] The function $(x,y)\rightarrow g(x,y)$ is concave;
  \item[(ii)] The function $\mathcal{H}(t)=\sup_{\pi\in\mathcal{A}}\mathcal{H}(t,X(t),Y(t),\pi,A_1^*(t),A_2^*(t),B_1^*(t),B_2^*(t))$ is concave and
  $$
\mathbb{E}\Bigl[\mathcal{H}(t,X,Y,\pi^*,A_1^*,A_2^*,B_1^*,B_2^*)\mid\mathcal{E}_t\Bigl]\,=\,\sup_{\pi\in\mathcal{A}} \mathbb{E}\Bigl[\mathcal{H}(t,X,Y,\pi,A_1^*,A_2^*,B_1^*,B_2^*)\mid\mathcal{E}_t\Bigl]\,.
$$

\end{itemize}
Furthermore, we assume the following:
$$
\mathbb{E}\Bigl[\int_0^T(X^*(t))^2\Bigl((B_1^*(t))^2+(B_2^*(t))^2\Bigl)dt\Bigl]<\infty\,;
$$
$$
\mathbb{E}\Bigl[\int_0^T(Y(t))^2\Bigl((B_3^*(t))^2+(B_4^*(t))^2\Bigl)dt\Bigl]<\infty\,;
$$

\begin{eqnarray*}
\mathbb{E}\Bigl[\int_0^T\Bigl\{(A_1^*(t))^2\Bigl((\sigma(t,X(t),Y(t),\pi(t)))^2+(\beta(t,X(t),Y(t),\pi(t)))^2\Bigl) && \\
 +(A_2^*(t))^2(\phi(Y(t)))^2\Bigl]dt\Bigl] &<& \infty\,,
\end{eqnarray*}
for all $\pi\in\mathcal{A}$.

Then, $\pi^*\in\mathcal{A}$ is an optimal strategy with the corresponding optimal state process $X^*$.
\end{theorem}

\begin{proof}

Let $\pi\in\mathcal{A}$ be an admissible strategy and $X(t)$ the corresponding wealth process. Then, following Framstad {\it et. al.} \cite{framstard}, Theorem 2.1., we have:

\begin{eqnarray*}
 \mathcal{J}(\pi^*)-\mathcal{J}(\pi)
  &=& \mathbb{E}\Bigl[\int_0^T(f(t,X^*(t),Y^*(t),\pi^*(t))-f(t,X(t),Y(t),\pi(t)))dt \\
  && +(g(X^*(T),Y^*(T))-g(X(T),Y(T))) \Bigl] \\
  &=& \mathcal{J}_1+\mathcal{J}_2\,.
\end{eqnarray*}
By condition $(i)$ and the integration by parts rule (Oksendal and Sulem \cite{Oksendal}, Lemma 3.6.), we have

\begin{eqnarray*}
  \mathcal{J}_2 &=& \mathbb{E}\Bigl[g(X^*(T),Y^*(T))-g(X(T),Y(T))\Bigl] \\
  &\geq& \mathbb{E}\Bigl[(X^*(T)-X(T))A^*_1(T)+(Y^*(T)-Y(T))A^*_2(T)\Bigl]  \\
   &=& \mathbb{E}\Bigl[\int_0^T(X^*(t)-X(t))dA^*_1(t)+\int_0^TA^*_1(t)(dX^*(t)-dX(t)) \\
   && +\int_0^T(Y^*(t)-Y(t))dA^*_2(t)+\int_0^TA^*_2(t)(dY^*(t)-dY(t)) \\
   && +\int_0^T[(\sigma(t,X^*(t),Y^*(t),\pi^*(t))- \sigma(t,X(t),Y(t),\pi(t)))B^*_1(t) \\
  &&  +  (\beta(t,X^*(t),Y^*(t),\pi^*(t))- \sigma(t,X(t),Y(t),\pi(t)))B^*_2(t)]dt \\
   &&  +\int_0^T(\phi(Y^*(t))-\phi(Y(t)))B_3^*(t)dt\Bigl] \\
   &=& \mathbb{E}\Bigl[-\int_0^T(X^*(t)-X(t))\frac{\partial \mathcal{H}^*}{\partial x}(t)dt -\int_0^T(Y^*(t)-Y(t))\frac{\partial \mathcal{H}^*}{\partial y}(t)dt \\
   && + \int_0^T(A_1^*(t)b(t,X^*(t),Y^*(t),\pi^*(t))-b(t,X(t),Y(t),\pi(t)))dt \\
   && + \int_0^T(\varphi(Y^*(t))-\varphi(Y(t)))A_2^*(t)dt +\int_0^T(\phi(Y^*(t))-\phi(Y(t)))B_3^*(t)dt \\
   && +\int_0^T[(\sigma(t,X^*(t),Y^*(t),\pi^*(t))- \sigma(t,X(t),Y(t),\pi(t)))B^*_1(t) \\
  &&  +  (\beta(t,X^*(t),Y^*(t),\pi^*(t))- \sigma(t,X(t),Y(t),\pi(t)))B^*_2(t)]dt\Bigl]\,,
\end{eqnarray*}
where we have used the notation $$\mathcal{H}^*(t)=\mathcal{H}(t,X^*(t),Y^*(t),\pi^*(t),A_1^*(t),A_2^*(t),B_1^*(t),B_2^*(t),B_3^*(t))\,.$$

On the other hand, by definition of $\mathcal{H}$ in \eqref{hamiltoniangeneral}, we see that

\begin{eqnarray*}
 \mathcal{J}_1 &=& \mathbb{E}\Bigl[\int_0^T(f(t,X^*(t),Y^*(t),\pi^*(t))-f(t,X(t),Y(t),\pi(t)))dt\Bigl] \\
   &=& \mathbb{E}\Bigl[\int_0^T [\mathcal{H}(t,X^*(t),Y^*(t),\pi^*(t),A_1^*(t),A_2^*(t),B_1^*(t),B_2^*(t),B_3^*(t))  \\
   && - \mathcal{H}(t,X^*(t),Y^*(t),\pi^*(t),A_1^*(t),A_2^*(t),B_1^*(t),B_2^*(t),B_3^*(t))]dt  \\
   && -\int_0^TA_1^*(t)(A_1^*(t)b(t,X^*(t),Y^*(t),\pi^*(t))-b(t,X(t),Y(t),\pi(t)))dt \\
   && - \int_0^T(\varphi(Y^*(t))-\varphi(Y(t)))A_2^*(t)dt +\int_0^T(\phi(Y^*(t))-\phi(Y(t)))B_3^*(t)dt \\
   && -\int_0^T[(\sigma(t,X^*(t),Y^*(t),\pi^*(t))- \sigma(t,X(t),Y(t),\pi(t)))B^*_1(t) \\
  &&  -  (\beta(t,X^*(t),Y^*(t),\pi^*(t))- \sigma(t,X(t),Y(t),\pi(t)))B^*_2(t)]dt\Bigl]\,.
\end{eqnarray*}
Then, summing the above two expressions, we obtain
\begin{eqnarray*}
 && \mathcal{J}_1+\mathcal{J}_2 \\
  &=&  \mathbb{E}\Bigl[\int_0^T [\mathcal{H}(t,X^*(t),Y^*(t),\pi^*(t),A_1^*(t),A_2^*(t),B_1^*(t),B_2^*(t),B_3^*(t))  \\
   && - \mathcal{H}(t,X(t),Y(t),\pi(t),A_1^*(t),A_2^*(t),B_1^*(t),B_2^*(t),B_3^*(t))]dt \\
   &&  -\int_0^T(X^*(t)-X(t))\frac{\partial \mathcal{H}^*}{\partial x}(t)dt -\int_0^T(Y^*(t)-Y(t))\frac{\partial \mathcal{H}^*}{\partial y}(t)dt\,.
\end{eqnarray*}

By the concavity of $\mathcal{H}$, i.e., conditions $(i)$ and $(ii)$, we have
\begin{eqnarray*}
   && \mathbb{E}\Bigl[\int_0^T [\mathcal{H}(t,X^*(t),Y^*(t),\pi^*(t),A_1^*(t),A_2^*(t),B_1^*(t),B_2^*(t),B_3^*(t))  \\
   && - \mathcal{H}(t,X(t),Y(t),\pi(t),A_1^*(t),A_2^*(t),B_1^*(t),B_2^*(t),B_3^*(t))]dt\mid\mathcal{E}_t\Bigl]  \\
   &\geq&  \mathbb{E}\Bigl[\int_0^T(X^*(t)-X(t))\frac{\partial \mathcal{H}^*}{\partial x}(t)dt +\int_0^T(Y^*(t)-Y(t))\frac{\partial \mathcal{H}^*}{\partial y}(t)dt \\
   && +\int_0^T(\pi^*(t)-\pi(t))\frac{\partial \mathcal{H}^*}{\partial \pi}(t)dt\mid\mathcal{E}_t\Bigl]\,.
\end{eqnarray*}

Then, by the maximality of the strategy $\pi^*\in\{\mathcal{E}_t\}$-measurable and the concavity of the Hamiltonian $\mathcal{H}$,
\begin{eqnarray*}
   && \mathbb{E}\Bigl[\int_0^T [\mathcal{H}(t,X^*(t),Y^*(t),\pi^*(t),A_1^*(t),A_2^*(t),B_1^*(t),B_2^*(t),B_3^*(t))  \\
   && - \mathcal{H}(t,X(t),Y(t),\pi(t),A_1^*(t),A_2^*(t),B_1^*(t),B_2^*(t),B_3^*(t))]dt\Bigl]  \\
   &\geq&  \mathbb{E}\Bigl[\int_0^T(X^*(t)-X(t))\frac{\partial \mathcal{H}^*}{\partial x}(t)dt +\int_0^T(Y^*(t)-Y(t))\frac{\partial \mathcal{H}^*}{\partial y}(t)dt\Bigl]\,.
\end{eqnarray*}
Hence  $\mathcal{J}(\pi^*)-\mathcal{J}(\pi)=\mathcal{J}_1+\mathcal{J}_2\geq0$. Therefore, $\mathcal{J}(\pi^*)\geq \mathcal{J}(\pi)$, for any strategy $\pi\in\mathcal{A}$. Then $\pi^*\in\mathcal{A}$ is optimal.
\end{proof}

\end{document}